\def\year{2018}\relax
\documentclass[letterpaper]{article} 
\usepackage{aaai18}  
\usepackage{times}
\usepackage{helvet}
\usepackage{courier}
\usepackage[table]{xcolor}
\usepackage{amssymb}
\usepackage{amsmath}
\usepackage{amsthm}
\usepackage{graphicx}
\usepackage{subcaption}
\usepackage[symbol]{footmisc}
\usepackage[colorinlistoftodos,prependcaption,textsize=tiny]{todonotes}
\usepackage{soul}
\usepackage{glossaries}
\usepackage{array}
\usepackage{multirow}
\usepackage{adjustbox}
\usepackage{rotating}
\usepackage{hyperref}
\usepackage{tabularx}
\usepackage{tikz}

\newacronym{vrcfr}{VR-MCCFR}{Variance Reduction MCCFR}

\newtheorem{theorem}{Theorem}
\newtheorem{lemma}{Lemma}

\newcommand{\E}{{ \mathop{\mathbb{E}} }}

\newcommand{\bE}{\mathbb{E}}

\newcommand{\bX}{\mathbf{X}}
\newcommand{\bY}{\mathbf{Y}}

\newcommand{\cA}{\mathcal{A}}

\newcommand{\cI}{\mathcal{I}}

\newcommand{\cH}{\mathcal{H}}

\newcommand{\cN}{\mathcal{N}}

\newcommand{\cZ}{\mathcal{Z}}

\newcommand{\defword}[1]{\textbf{\boldmath{#1}}}

\newcommand{\Var}{\mathbb{V}\text{ar}}
\newcommand{\Cov}{\mathbb{C}\text{ov}}

\definecolor{darkgreen}{RGB}{0,125,0}

\newcounter{mlNoteCounter}

\newcounter{rkNoteCounter}

\begin{document}

\title{Variance Reduction in Monte Carlo Counterfactual Regret Minimization (VR-MCCFR) for Extensive Form Games using Baselines}
\author{Martin Schmid$^{1}$, Neil Burch$^{1}$, Marc Lanctot$^{1}$, Matej Moravcik$^{1}$, Rudolf Kadlec$^{1}$, Michael Bowling$^{1, 2}$ \\
DeepMind$^1$ \\
University of Alberta$^2$ \\
\{mschmid,burchn,lanctot,moravcik,rudolfkadlec,bowlingm\}@google.com
}
\maketitle

\begin{abstract}
Learning strategies for imperfect information games from samples of interaction is a challenging problem. A common method for this setting,  Monte Carlo Counterfactual Regret Minimization (MCCFR), can have slow long-term convergence rates due to high variance. In this paper, we introduce a variance reduction technique (VR-MCCFR) that applies to any sampling variant of MCCFR. Using this technique, per-iteration estimated values and updates are reformulated as a function of sampled values and state-action baselines, similar to their use in policy gradient reinforcement learning.
The new formulation allows estimates to be bootstrapped from other estimates within the same episode, propagating the benefits of baselines along the sampled trajectory; the estimates remain unbiased even when bootstrapping from other estimates. Finally, we show that given a perfect baseline, the variance of the value estimates can be reduced to zero.
Experimental evaluation shows that VR-MCCFR brings an order of magnitude speedup, while the empirical variance decreases by three orders of magnitude.
The decreased variance allows for the first time CFR+ to be used with sampling, increasing the speedup to two orders of magnitude.

\end{abstract}

%

\section{Introduction}

Policy gradient algorithms have shown remarkable success in single-agent reinforcement learning (RL)~\cite{Mnih2016asynchronous,schulman2017proximal}. While there has been evidence
of empirical success in multiagent problems~\cite{Foerster18,Bansal18}, the assumptions made by RL methods generally do not hold in multiagent partially-observable environments. Hence, they are not guaranteed to find an optimal policy, even with tabular representations in two-player zero-sum (competitive) games~\cite{Littman94markovgames}.
As a result, policy iteration algorithms based on computational game theory and regret minimization have been the preferred formalism in this setting. 
Counterfactual regret minimization~\cite{CFR} has been a core component of this progress in Poker AI, leading to solving Heads-Up Limit Texas Hold'em~\cite{Bowling15Poker} and defeating professional poker players in No-Limit~\cite{Moravcik17DeepStack,Brown17Libratus}.

\begin{figure}[t!]
\includegraphics[width=0.99\hsize]{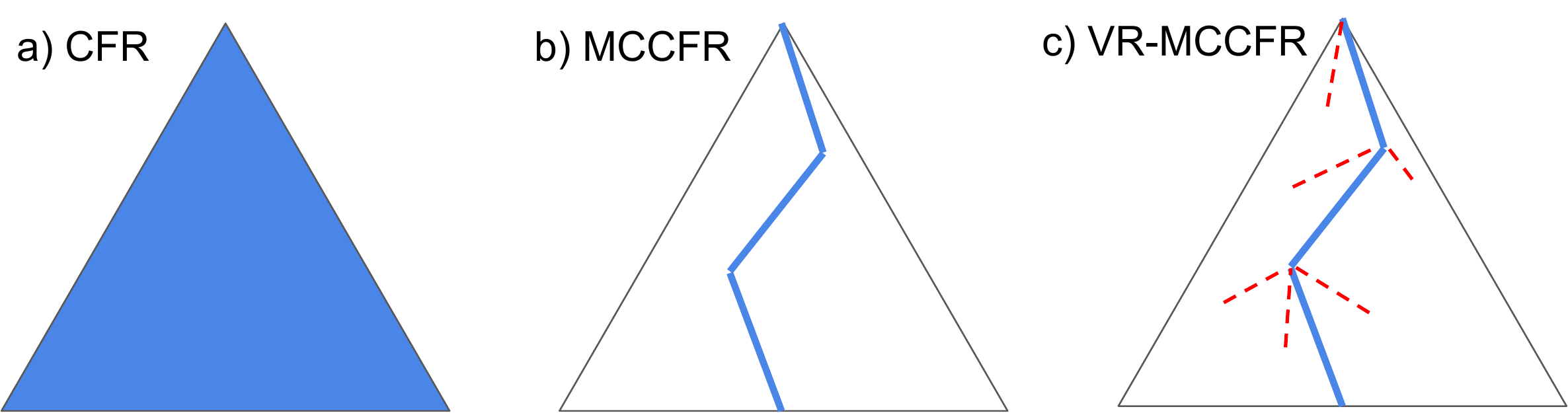}
\caption{High-level overview of \gls{vrcfr} and related methods. 
a) CFR traverses the entire tree on every iteration.
b) MCCFR samples trajectories and computes the values only for the sampled actions, while the off-trajectory actions are treated as zero-valued.
While MCCFR uses importance sampling weight to ensure the values are unbiased, the sampling introduces high variance.
c) \gls{vrcfr}  follows the same sampling framework as MCCFR, but uses baseline values for both sampled actions (in blue) as well as the off-trajectory actions (in red).
These baselines use control variates and send up bootstrapped estimates to decrease the per-iteration variance thus speeding up the convergence.
}
\label{fig:high_level}
\end{figure}

The two fields of RL and computational game theory have largely grown independently. However, there has been recent work that relates approaches within these two communities. Fictitious self-play uses RL to compute approximate best responses and supervised learning to combine responses~\cite{Heinrich15FSP}. This idea is extended to a unified training framework that can produce more general policies by regularizing over generated response oracles~\cite{Lanctot17PSRO}.
RL-style regressors were first used to compress regrets in game theorietic algorithms~\cite{Waugh15solving}.
DeepStack introduced deep neural networks as generalized value-function approximators for online planning in imperfect information games~\cite{Moravcik17DeepStack}. These value functions operate on a belief-space over all possible states consistent with the players' observations.

This paper similarly unites concepts from both fields, proposing an unbiased variance reduction technique for Monte Carlo counterfactual regret minimization using an analog of state-action baselines from actor-critic RL methods.
While policy gradient methods typically involve Monte Carlo estimates, the analog in imperfect information settings is Monte Carlo Counterfactual Regret Minimization (MCCFR)~\cite{Lanctot09mccfr}.
Policy gradient estimates based on a single sample of an episode suffer significantly from variance. 
A common technique to decrease the variance is a state or state-action dependent baseline value that is subtracted from the observed return.
These methods can drastically improve the convergence speed.
However, no such methods are known for MCCFR.

MCCFR is a sample based algorithm in imperfect information settings, which approximates counterfactual regret minimization (CFR) by estimating regret quantities necessary for updating the policy.
While MCCFR can offer faster short-term convergence than original CFR in large games, it suffers from high variance which leads to slower long-term convergence.

CFR+ provides significantly faster empirical performance and made solving Heads-Up Limit Texas Hold'em possible~\cite{Bowling15Poker}.
Unfortunately, CFR+ has so far did not outperform CFR in Monte Carlo settings \cite{burch2017time} (also see Figure (~\ref{fig:mccfr_vs_mccfr_plus}) in the appendix for an experiment).

In this work, we reformulate the value estimates using a control variate and a state-action baseline. The new formulation includes any approximation of the counterfactual values, which allows for a range of different ways to insert domain-specific knowledge (if available) but also to design values that are learned online.

Our experiments show two orders of magnitude improvement over MCCFR.
For the common testbed imperfect information game -- Leduc Poker -- \gls{vrcfr} with a state-action baseline needs 250 times fewer iterations than MCCFR to reach the same solution quality.
In contrast to RL algorithms in perfect information settings, where state-action baselines bring little to no improvement over state baselines~\cite{tucker2018mirage}, state-action baselines lead to significant improvement over state baselines in multiagent partially-observable settings.
We suspect this is due to variance from the environment and different dynamics of the policies during the computation.

\section{Related Work}

There are standard variance reduction techniques for Monte Carlo sampling methods~\cite{Owen13MCBook} and the use of control variates in these settings has a long history~\cite{boyle1977options}. 
Reducing variance is particularly important when estimating gradients from sample trajectories. 
Consequentially, the use of a control variates using baseline has become standard practice in policy gradient methods~\cite{Williams92,SuttonBarto17}. In RL, action-dependent baselines have recently shown promise~\cite{Wu18,liu2018action} but the degree to which variance is indeed reduced remains unclear~\cite{tucker2018mirage}. We show that in our setting of MCCFR in imperfect information multiplayer games, action-dependent baselines necessarily influence the variance of the estimates, and we confirm the reduction empirically. This is important because lower-variance estimates lead to better regret bounds~\cite{Gibson12probing}.

There have been a few uses of variance reduction techniques in multiplayer games, within Monte Carlo tree search (MCTS). In MCTS, control variates have used to augment the reward along a trajectory using a property of the state before and after a transition~\cite{Veness11variance} and to augment the outcome of a rollout from its length or some pre-determined quality of the states visited~\cite{Pepels14Quality}.

Our baseline-improved estimates are similar to the ones used in AIVAT~\cite{AAAI1817316}. AIVAT defines estimates of expected values using heuristic values of states as baselines in practice. Unlike this work, AIVAT was only used for evaluation of strategies.

To the best of our knowledge, there has been two applications of variance reduction in Monte Carlo CFR: by manipulating the chance node distribution~\cite[Section 7.5]{Lanctot13phdthesis} and by sampling (``probing'') more trajectories for more estimates of the underlying values~\cite{Gibson12probing}. The variance reduction (and resulting drop in convergence rate) is modest in both cases, whereas we show more than a two order of magnitude speed-up in convergence using our method.

\section{Background}

We start with the formal background necessary to understand our method. For details, see~\cite{ShohamLB09,SuttonBarto17}.

A two player \defword{extensive-form game} is tuple $(\cN, \cA, \cH, \cZ, \tau, u, \cI)$. 

$\cN = \{ 1, 2, c \}$ is a finite set of players, where $c$ is a special player called chance. $\cA$ is a finite set of actions. Players take turns choosing actions, which are composed into sequences called {\it histories}; the set of all valid histories is $\cH$, and the set of all terminal histories (games) is $\cZ \subseteq \cH$.
We use the notation $h' \sqsubseteq h$ to mean that $h'$ is a prefix sequence or equal to $h$.
Given a nonterminal history $h$, the player function $\tau : \cH \setminus \cZ \rightarrow \cN$ determines who acts at $h$. The utility function $u : (\cN \setminus \{ c \}) \times \mathcal{Z} \rightarrow [u_{\min}, u_{\max}] \subset \mathbb{R}$ assigns a payoff to each player for each terminal history $z \in \cZ$.

The notion of a state in imperfect information games requires groupings of histories: $\cI_i$ for some player $i \in \cN$ is a partition of $\{ h \in \cH~|~\tau(h) = i \}$ into parts $I \in \cI_i$ such that $h, h' \in I$ if player $i$ cannot distinguish $h$ from $h'$ given the information known to player $i$ at the two histories. We call these \defword{information sets}. For example, in Texas Hold'em poker, for all $I \in \cI_i$, the (public) actions are the same for all $h, h' \in I$, and $h$ only differs from $h'$ in cards dealt to the opponents (actions chosen by chance). For convenience, we refer to $I(h)$ as the information state that contains $h$.

At any $I$, there is a subset of legal actions $A(I) \subseteq \cA$.
To choose actions, each player $i$ uses a \defword{strategy} $\sigma_i: I \rightarrow \Delta(A(I))$, where $\Delta(X)$ refers to the set of probability distributions over $X$.
We use the shorthand $\sigma(h,a)$ to refer to $\sigma(I(h),a)$.
Given some history $h$, we define the \defword{reach probability} $\pi^\sigma(h) = \Pi_{h'a \sqsubset h} \sigma_{\tau(h')}(I(h'), a)$ to be the product of all action probabilities leading up to $h$. This reach probability contains all players' actions, but can be separated $\pi^\sigma(h) = \pi_i^\sigma(h) \pi_{-i}^\sigma(h)$ into player $i$'s actions' contribution and the contribution of the opponents' of player $i$ (including chance).

Finally, it is often useful to consider the \defword{augmented information sets}~\cite{Burch14CFRD}.
While an information set $I$  groups histories $h$ that player $i=\tau(h)$ cannot distinguish, an augmented information set groups histories that player $i$ can not distinguish, including these where $\tau(h) \neq i$.
For a history $h$, we denote an augmented information set of player $i$ as $I_i(h)$.
Note that the if $\tau(h) = i$ then $I_i(h) = I(h)$ and  $I(h) =  I_{\tau(h)}(h)$.

\subsection{Counterfactual Regret Minimization}

Counterfactual Regret (CFR) Minimization is an iterative algorithm that produces a sequence of strategies $\sigma^0, \sigma^1, \ldots, \sigma^T$, whose average strategy $\bar{\sigma}^T$ converges to an approximate Nash equilibrium as $T \rightarrow \infty$ in two-player zero-sum games~\cite{CFR}.
Specifically, on iteration $t$, for each $I$, it computes \defword{counterfactual values}.
Define $\cZ_I = \{ (h, z) \in \cH \times \cZ~|~ h \in I, h \sqsubseteq z\}$, and $u_i^{\sigma^t}(h,z) = \pi^{\sigma^t}(h, z) u_i(z)$. We will also sometimes use the short form $u^{\sigma}_i(h) = \sum_{z \in \cZ, h \sqsubseteq z} u^{\sigma}_i(h,z)$. A counterfactual value is:
\begin{align}
\label{eq:cfv}
v_i(\sigma^t, I) = \sum_{(h,z) \in \cZ_I} \pi^{\sigma^t}_{-i} (h) u_i^{\sigma^t}(h, z).
\end{align}
We also define an action-dependent counterfactual value,
\begin{align}
\label{eq:action-dep-cfv}
v_i(\sigma, I, a) = \sum_{(h,z) \in \cZ_I} \pi^{\sigma}_{-i}(ha) u^{\sigma}(ha,z),     
\end{align}
where $ha$ is the sequence $h$ followed by the action $a$.
The values are analogous to the difference in $Q$-values and $V$-values in RL, and indeed we
have $v_i(\sigma, I) = \sum_a \sigma(I,a) v_i(\sigma, I, a)$.
CFR then computes a \defword{counterfactual regret} for {\it not taking} $a$ at $I$:
\begin{align}
\label{eq:regret}
r^t(I, a) = v_i(\sigma^t, I, a) - v_i(\sigma^t, I),
\end{align}
This regret is then accumulated $R^T(I,a) = \sum_{t=1}^T r^t(I,a)$, which is used to update the strategies using \defword{regret-matching}~\cite{Hart00}:
\begin{align} \sigma^{T+1}(I,a) = \frac{(R^{T}(I,a))^+}{\sum_{a \in A(I)}
\label{eq:rm}
(R^{T}(I,a))^+},
\end{align}
where $(x)^+ = \max(x, 0)$,
or to the uniform strategy if $\sum_a (R^T(I,a))^+ = 0$.
CFR+ works by thresholding the quantity at each round~\cite{Tammelin15CFRPlus}: define $Q^0(I,a) = 0$ and $Q^T(I,a) = (Q^{T-1} + r^T(I,a))^+$; CFR+ updates the policy by replacing $R^T$ by $Q^T$  in equation~\ref{eq:rm}. In addition, it always alternates the regret updates of the players (whereas some variants of CFR update both players), and the average strategy places more (linearly increasing) weight on more recent iterations.

If for player $i$ we denote $u(\sigma) = u_i(\sigma_i, \sigma_{-i})$, and run CFR for $T$ iterations, then we can define the \defword{overall regret} of the strategies produced as:
\[
R_i^{T} = \max_{\sigma'_i} \sum_{t=1}^{T} \left( v_i(\sigma'_i, \sigma_{-i}^t) - v_i(\sigma^t) \right).
\]
CFR ensures that $R_i^T/T \rightarrow 0$ as $T \rightarrow \infty$. When two players minimize regret, the folk theorem then guarantees a bound on the distance to a Nash equilibrium as a function of $R_i^T/T$.

To compute $v_i$ precisely, each iteration requires traversing over subtrees under each $a \in A(I)$ at each $I$.
Next, we describe variants that allow sampling parts of the trees and using estimates of these quantities.

\subsection{Monte Carlo CFR}
Monte Carlo CFR (MCCFR) introduces sample estimates of the counterfactual values, by visiting and updating quantities over only part of the entire tree. MCCFR is a general family of algorithms: each instance defined by a specific sampling policy. For ease of exposition and to show the similarity to RL, we focus on \defword{outcome sampling}~\cite{Lanctot09mccfr}; however, our baseline-enhanced estimates can be used in all MCCFR variants.
A \defword{sampling policy} $\xi$ is defined in the same way as a strategy (a distribution over $A(I)$ for all $I$) with a restriction that $\xi(h,a) > 0$ for all histories and actions. Given a terminal history sampled with probability $q(z) = \pi^{\xi}(z)$, a \defword{sampled counterfactual value} $\tilde{v}_i(\sigma, I | z)$
\begin{align}
\label{eq:sampled-cfv}
= \tilde{v}_i(\sigma, h | z) = \frac{ \pi^{\sigma}_{-i}(h) u_i^{\sigma}(h,z)}{q(z)}, \mbox{ for } h \in I, h \sqsubseteq z, 
\end{align}
and $0$ for histories that were not played, $h \not\sqsubseteq z$.
The estimate is unbiased:
$\E_{z \sim \xi}[\tilde{v}_i(\sigma, I | z)] = v_i(\sigma, I)$, by~\cite[Lemma 1]{Lanctot09mccfr}.
As a result, $\tilde{v}_i$ can be used in Equation~\ref{eq:regret} to accumulate estimated regrets $\tilde{r}^t(I,a) = \tilde{v}_i(\sigma^t, I, a) - \tilde{v}_i(\sigma^t, I)$ instead.
The regret bound requires an additional term $\frac{1}{\min_{z \in \cZ}q(z)}$, which is exponential in the length of $z$ and similar observations have been made in RL~\cite{RUDDER}.
The main problem with the sampling variants is that they introduce variance that can have a significant effect on long-term convergence~\cite{Gibson12probing}. 

\subsection{Control Variates}

Suppose one is trying to estimate a statistic of a random variable, $X$, such as its mean, from samples $\bX = (X_1, X_2, \cdots, X_n)$. A crude Monte Carlo estimator is defined to be $\hat{X}^{mc} = \frac{1}{n} \sum_{i = 1}^n X_i$.
A \defword{control variate} is a random variable $Y$ with a known mean $\mu_Y = \bE[Y]$, that is paired with the original variable, such that samples are instead  of the form $(\bX, \bY)$~\cite{Owen13MCBook}. A new random variable is then defined, $Z_i = X_i + c (Y_i - \mu_Y)$.
An estimator $\hat{Z}^{cv} = \frac{1}{n} \sum_{i=1}^{n} Z_i$. Since $\bE[Z_i] = \bE[X_i]$ for any value of $c$, $\hat{Z}^{cv}$ can be used in place of $\hat{X}^{mc}$. 
with variance $\Var[Z_i] = \Var[X_i] + c^2 \Var[Y_i] + 2c \Cov[X_i,Y_i]$. So when $X$ and $Y$ are positively correlated and $c < 0$, variance is reduced when $\Cov[X,Y] > \frac{c^2}{2}\Var[Y]$.

\subsection{Reinforcement Learning Mapping}

There are several analogies to make between Monte Carlo CFR in imperfect information games and reinforcement learning. Since our technique builds on ideas that have been widely used in RL, we end the background by providing a small discussion of the links.

First, dynamics of an imperfect information game are similar to  a partially-observable episodic MDP without any cycles. 
Policies and strategies are identically defined, but in imperfect information games a deterministic optimal (Nash) strategy may not exist causing most of the RL methods to fail to converge. 
The search for a minmax-optimal strategy with several players is the main reason CFR is used instead of, for example, value iteration. However, both operate by defining values of states which are analogous (counterfactual values versus expected values) since they are both functions of the strategy/policy; therefore, can be viewed as a kind of policy iteration which computes the values and from which a policy is derived. However, the iterates $\sigma^t$ are not guaranteed to converge to the optimal strategy, only the average strategy $\bar{\sigma}^t$ does.

Monte Carlo CFR is an off-policy Monte Carlo analog. The value estimates are unbiased specifically because they are corrected by importance sampling. Most applications of MCCFR have operated with tabular representations, but this is mostly due to the differences in objectives. Function approximation methods have been proposed for CFR~\cite{Waugh15solving} but the variance from pure Monte Carlo methods may prevent such techniques in MCCFR. The use of baselines has been widely successful in policy gradient methods, so reducing the variance could enable the practical use of function approximation in MCCFR.

\section{Monte Carlo CFR with Baselines}
\label{sec:mccfrb}

We now introduce our technique: MCCFR with baselines. While the baselines are analogous to those from policy gradient methods (using counterfactual values), there are slight differences in their construction.

Our technique constructs value estimates using control variates.
Note that MCCFR is using sampled estimates of counterfactual values $\tilde{v}_i(\sigma, I)$ whose expected value is the counterfactual value $v_i(\sigma, I)$.
First, we introduce an \defword{estimated counterfactual} value $\hat{v}_i(\sigma, I)$ to be any estimator of the counterfactual value (not necessarily $\tilde{v}_i$ as defined above, but this is one possibility).

We now define an action-dependent \defword{baseline} $b_i(I, a)$ that, as in RL, serves as a basis for the sampled values.
The intent is to define a baseline function to approximate or be correlated with $\bE[\hat{v}_i(\sigma, I, a)]$.
We also define a sampled baseline $\hat{b}_i(I, a)$ as an estimator such that $\bE[\hat{b}_i(I, a)] = b_i(I, a)$.  
From this, we construct a new baseline-enhanced estimate for the counterfactual values:
\begin{align}
\label{eq:baseline-enhanced-v}
\widehat{v}^b_i(\sigma, I, a) = 
    \widehat{v}_i(\sigma, I, a) - \hat{b}_i(\sigma, I, a) +  b_i(\sigma, I, a)
\end{align}
First, note that $\hat{b}_i$ is a control variate with $c = -1$. Therefore, it is important that $\hat{b}_i$ be correlated with $\hat{v}_i$. 
The main idea of our technique is to replace $\tilde{v}_i(\sigma, I, a)$ with $\hat{v}^b_i(\sigma, I, a)$. 
A key property is that by doing so, the expectation remains unchanged.

\begin{figure*}
    \centering
    \begin{subfigure}[b]{0.34\textwidth}
        \includegraphics[width=\textwidth]{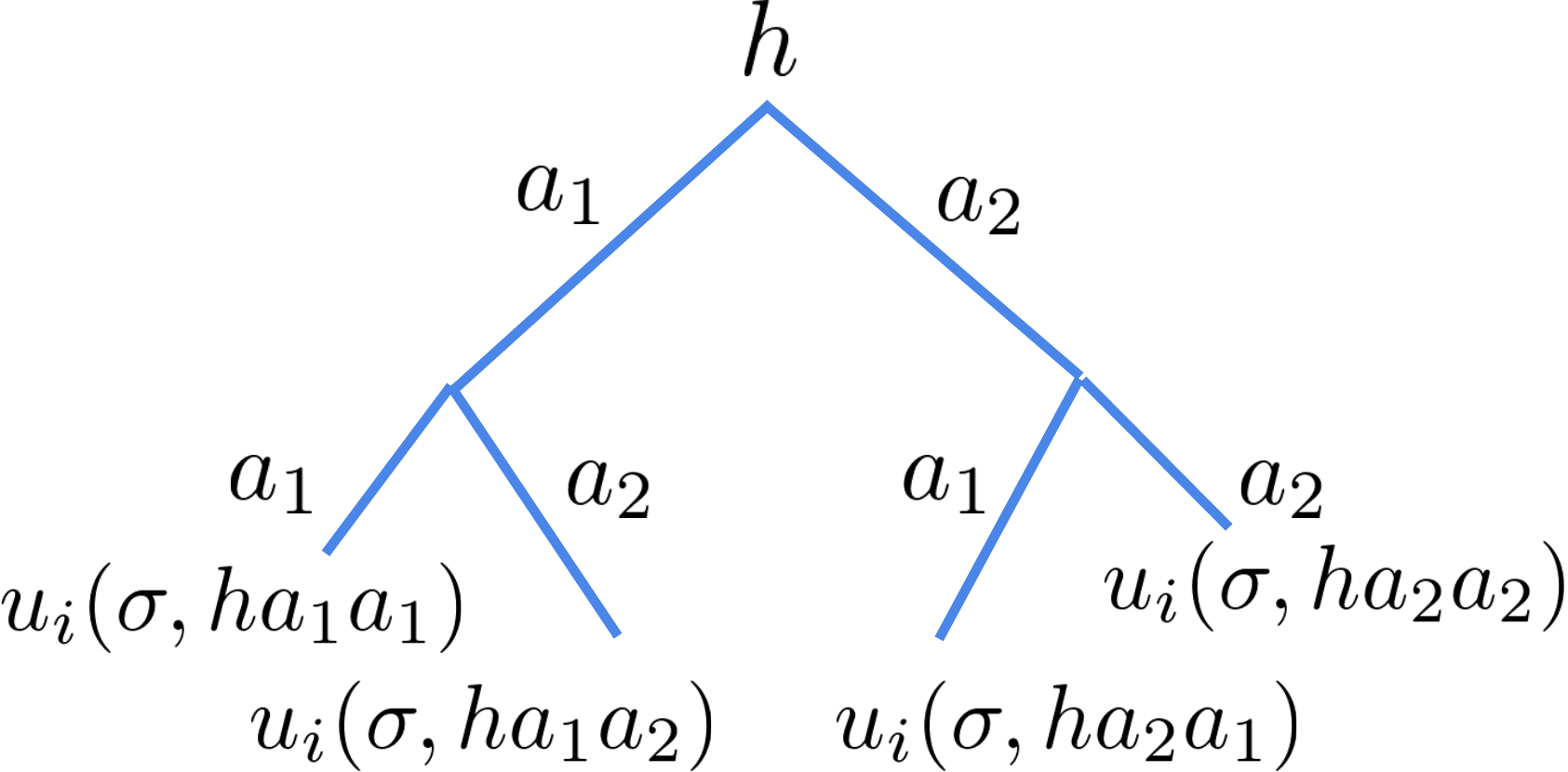}
        \caption{CFR}
        \label{fig:gull}
    \end{subfigure}
    \qquad
    \begin{subfigure}[b]{0.12\textwidth}
        \includegraphics[width=\textwidth]{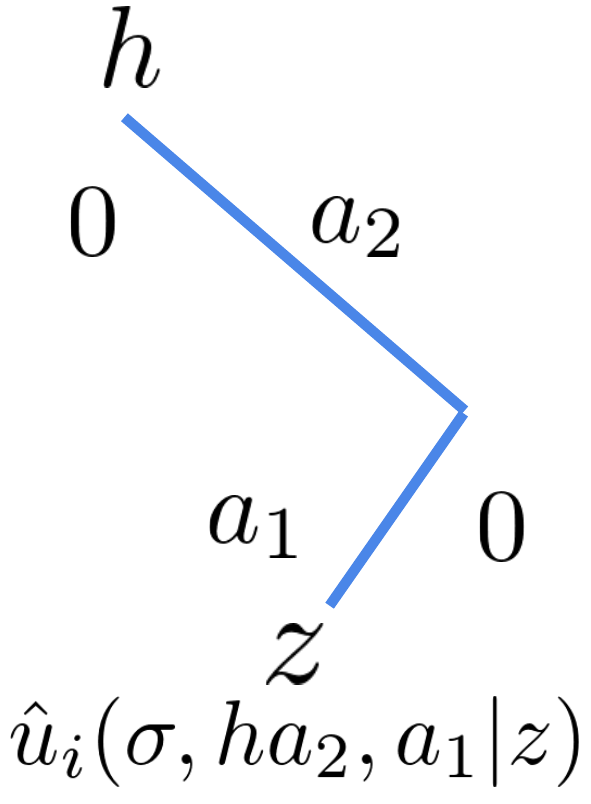}
        \caption{MCCFR}
        \label{fig:tiger}
    \end{subfigure}
    \qquad
    \begin{subfigure}[b]{0.30\textwidth}
        \includegraphics[width=\textwidth]{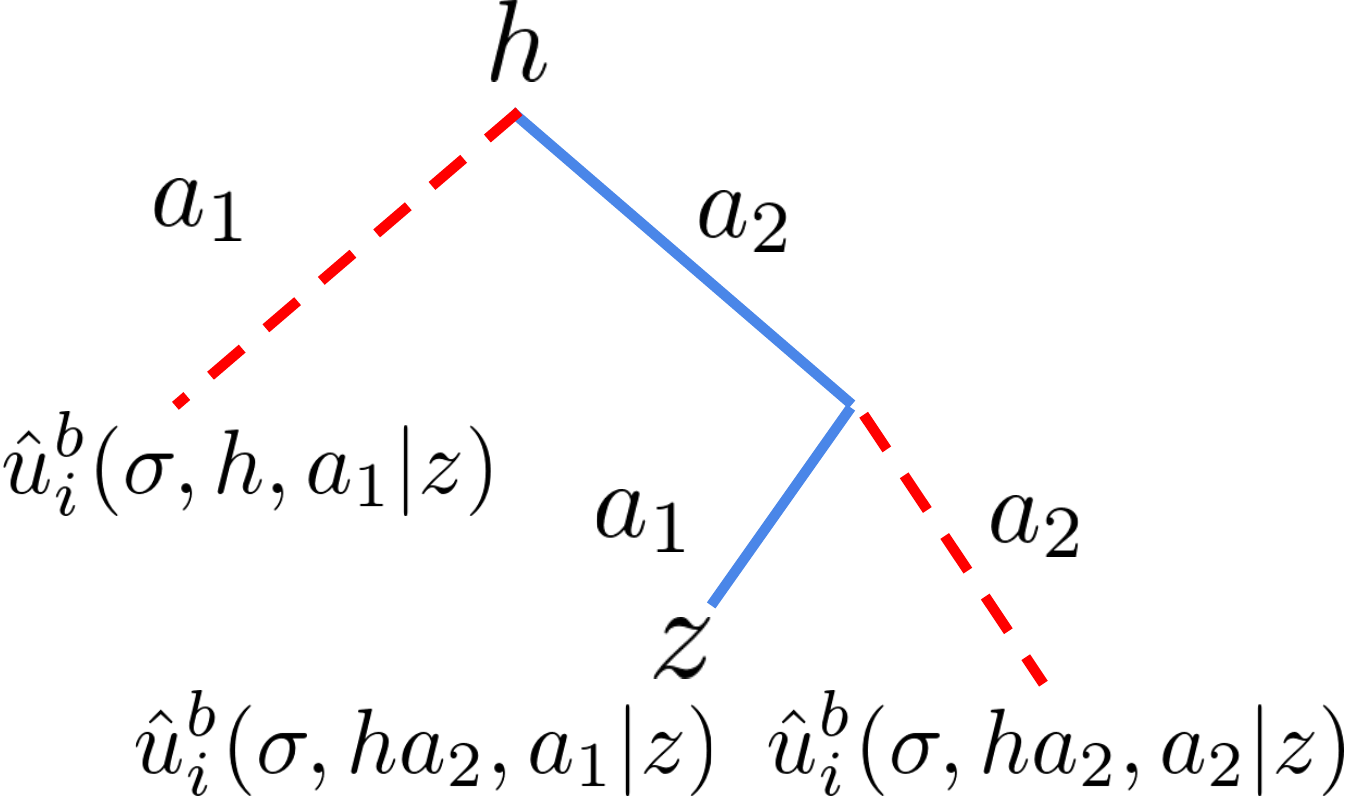}
        \caption{VR-MCCFR}
        \label{fig:mouse}
    \end{subfigure}
    \caption{Values and updates for the discussed methods: (a) CFR udpates the full tree and thus uses the exact values for all the actions,
    (b) MCCFR updates only a single path, and uses the sampled values for the sampled actions and zero values for the off-trajectory actions,
    (c) \gls{vrcfr} also updates only a single path, but uses the bootstrapped baseline-enhanced values for the sampled actions and baseline-enhanced values for the off-trajectory actions.
    }
    \label{fig:animals}
\end{figure*}

\begin{lemma}
\label{lemma:unbiasied}
For any $i \in \cN - \{c\}, \sigma_i, I \in \cI, a \in A(I)$, if $\bE[\hat{b}_i(I, a)] = b_i(I, a)$ and $\bE[\hat{v}_i(\sigma, I, a)] = v_i(\sigma, I, a)$, then $\bE[\hat{v}_i^b(\sigma, I,a)] = v_i(\sigma, I, a)$.
\end{lemma}

The proof is in the appendix. 
As a result, any baseline whose expectation is known can be used and the baseline-enhanced estimates are consistent. However, not all baselines will decrease variance. For example, if $\Cov[\hat{v}_i, \hat{b}_i]$ is too low, then the $\Var[\hat{b}_i]$ term could dominate and actually increase the variance.

\subsection{Recursive Bootstrapping}

Consider the individual computation (\ref{eq:cfv}) 
for all the information sets on the path to a sampled terminal history $z$.
Given that the counterfactual values up the tree can be computed from the counterfactual values down the tree, it is natural to consider propagating the already baseline-enhanced
counterfactual values (\ref{eq:baseline-enhanced-v}) rather than the original noisy sampled values - thus propagating the benefits up the tree.
The Lemma (\ref{lemma:bootstrapped-unbiased}) then shows that by doing so, the updates remain unbiased.
Our experimental section shows that such bootstrapping a crucial component for the proper performance of the method.

To properly formalize this bootstrapping computation, we must first recursively define the \defword{expected value}:

\begin{align}
\label{eq:bootstrapped-u}
\hat{u}_i(\sigma, h, a | z) = \left\{ \begin{array}{ll}
     \hat{u}_i(\sigma, ha|z) / {\xi(h,a)}   & \mbox{if $ha \sqsubseteq z$} \\
     0 & \mbox{otherwise}\end{array} \right.,
\end{align}
and 
\begin{align}
\label{eq:bootstrapped-u-history}
\hat{u}_i(\sigma, h | z) = \left\{ \begin{array}{ll}
u_i(h) & \mbox{if $h = z$} \\
\sum_{a} \sigma(h,a) \hat{u}_i(\sigma, h, a | z) & \mbox{if $h \sqsubset z$} \\
0 & \mbox{otherwise}\end{array}\right..
\end{align}

Next, we define a baseline-enhanced version of the expected value. 
Note that the baseline $b_i(I, a)$ can be arbitrary, but we discuss a particular choice and update of the baseline in the later section.
For every action, given a specific sampled trajectory $z$, then $\hat{u}^b_i(\sigma, h, a | z) =$

\begin{align}
\label{eq:bootstrapped-ub}
 \left\{ \begin{array}{ll}
 b_i(I_i(h), a) + \frac{\hat{u}^b_i(\sigma, ha | z) - b_i(I_i(h), a)}{\xi(h,a)} & \mbox{if $ha \sqsubseteq z$} \\
 b_i(I_i(h), a) & \mbox{if $h \sqsubset z$, $ha \not\sqsubseteq z$} \\
 0 & \mbox{otherwise}\end{array} \right.
\end{align}
and
\begin{align}
\label{eq:bootstrapped-ub-history}
\hat{u}^b_i(\sigma, h | z) = \left\{ \begin{array}{ll}
  u_i(h) & \mbox{if $h = z$} \\
  \sum_{a} \sigma(h,a) \hat{u}^b_i(\sigma, h, a | z) & \mbox{if $h \sqsubset z$} \\
  0 & \mbox{otherwise} \end{array}\right..
\end{align}

These are the values that are bootstrapped.
We estimate counterfactual values needed for the regret updates using these values as:

\begin{align}
\label{eq:bootstrapped-vb}
\hat{v}^b_i(\sigma, I(h), a | z) = \hat{v}^b_i(\sigma, h, a | z) = \frac{\pi^{\sigma}_{-i}(h)}{q(h)} \hat{u}^b_i(\sigma, h, a | z).
\end{align}

We can now formally state that the bootstrapping keeps the counterfactual values unbiased:

\begin{lemma}
\label{lemma:bootstrapped-unbiased}
Let $\hat{v}_i^b$ be defined as in Equation~\ref{eq:bootstrapped-vb}. Then, for any $i \in \cN - \{c\}, \sigma_i, I \in \cI, a \in A(I)$, it holds that $\bE_z[\hat{v}_i^b(\sigma, I,a | z)] = v_i(\sigma, I, a)$.
\end{lemma}

The proof is in the appendix. 
Since each estimate builds on other estimates, the benefit of the reduction in variance can be propagated up through the tree.

Another key result is that there exists a perfect baseline that leads to zero-variance estimates at the updated information sets.
\begin{lemma}
\label{lemma:zero-variance}
There exists a perfect baseline $b^*$ and optimal unbiased estimator $\hat{v}^*_i(\sigma, h, a)$ such that under a specific update scheme:  $\Var_{h,z \sim \xi, h \in I, h \sqsubseteq z}[\hat{v}^*_i(\sigma, h, a|z)] = 0$.
\end{lemma}
The proof and description of the update scheme are in the appendix. We will refer to $b^*$ as the \defword{oracle baseline}. Note that even when using the oracle baseline, the convergence rate of MCCFR is still not identical to CFR because each iteration applies regret updates to a portion of the tree, whereas CFR updates the entire tree.

Finally, using unbiased estimates to tabulate regrets $\hat{r}(I,a)$ for each $I$ and $a$ leads to a probabilistic regret bound:
\begin{theorem}{\cite[Theorem 2]{Gibson12probing}}
\label{theorem:gibson}
For some unbiased estimator of the counterfactual values $\hat{v}_i$ and a bound on the difference in its value $\hat{\Delta}_i = |\hat{v}_i(\sigma, I, a) - \hat{v}_i(\sigma, I, a')|$, with probability 1-$p$, $\frac{R_i^T}{T}$
\[
\leq \left( \hat{\Delta}_i + \frac{\sqrt{ \max_{t, I, a} \Var[r_i^t(I,a) - \hat{r}_i^t(I,a)] }}{\sqrt{p}} \right) \frac{|\cI_i| |\cA_i|}{\sqrt{T}}.
\]
\end{theorem}

\subsection{Choice of Baselines}

How does one choose a baseline, given that we want these to be good estimates of the individual counterfactual values?
A common choice of the baseline in policy gradient algorithms is the mean value of the state, which is learned online~\cite{Mnih2016asynchronous}. 
Inspired by this, we choose a similar quantity: the average expected value $\bar{\hat{u}}_i(I_i, a)$. That is, in addition to accumulating regret for each $I$, average expected values are also tracked.

While a direct average can be tracked, we found that an exponentially-decaying average that places heavier weight on more recent samples to be more effective in practice. On the $k^{th}$ visit to $I$ at iteration $t$, 
\[
\bar{\hat{u}}^k_i(I_i, a) =
\left\{ \begin{array}{ll}
0 & \mbox{if $k = 0$} \\
 (1-\alpha)\bar{\hat{u}}^{k-1}_i(I_i, a) + \alpha\hat{u}^b_i(\sigma^t, I_i, a) & \mbox{if $k>0$} \end{array} \right.
\]
We then define the baseline $b_i(I_i, a) = \bar{\hat{u}}_i(I_i, a)$, and
\[
\hat{b}_i(I_i, a | z) =
\left\{ \begin{array}{ll}
 b_i(I_i, a) / \xi(I_i,a) & \mbox{if $ha \sqsubseteq z, h \in I_i$ } \\
 0 & \mbox{otherwise}.\end{array} \right.
\]
The baseline can therefore be thought as {\it local} to $I_i$ since it depends only on quantities defined and tracked at $I_i$. Note that $\bE_{a \sim \xi(I_i)}[\hat{b}_i(I_i, a|z)] = b_i(I_i, a)$ as required.

\subsection{Summary of the Full Algorithm}

We now summarize the technique developed above. One iteration of the algorithm consists of:

\begin{enumerate}
\item Repeat the steps below for each $i \in \cN - \{ c \}$.
\item Sample a trajectory $z \sim \xi$.
\item For each history $h \sqsubseteq z$ in reverse order (longest first):
    \begin{enumerate}
    \item If $h$ is terminal, simply return $u_i(h)$
    \item Obtain current strategy $\sigma(I)$ from Eq.~\ref{eq:rm} using cumulative regrets $R(I,a)$ where $h \in I$.
    \item Use the child value $\hat{u}^b_i(\sigma, ha)$ to compute $\hat{u}_i^b(\sigma, h)$ as in Eq.~\ref{eq:bootstrapped-ub}.
    
    \item If $\tau(h) = i$  then for $a \in A(I)$, compute $\hat{v}_i^b(\sigma, I, a) = \frac{\pi_{-i}(h)}{q(h)} \hat{u}_i^b(\sigma, ha)$ and accumulate regrets $R(I,a) \leftarrow R(I,a) + \hat{v}_i^b(\sigma, I,a) - \hat{v}_i^b(\sigma, I)$. 
    \item Update $\bar{\hat{u}}(\sigma, I_i, a)$.
    
    \item Finally, return $\hat{u}^b_i(\sigma, h)$.
    \end{enumerate}
\end{enumerate}

Note that the original outcome sampling is an instance of this algorithm. Specifically, when $b_i(I_i, a) = 0$, then $\hat{v}^b_i(\sigma, I, a) = \tilde{v}_i(\sigma, I, a)$. Step by step example of the computation is in the appendix.

\section{Experimental Results}
 
We evaluate the performance of our method on \defword{Leduc poker}~\cite{05uai-bayes}, a commonly used benchmark poker game. Players have an unlimited number of chips, and the deck has six cards, divided into two suits of three identically-ranked cards. There are two rounds of betting; after the first round a single public card is revealed from the deck. Each player antes 1 chip to play, receiving one private card. There are at most two bet or raise actions per round, with a fixed size of 2 chips in the first round, and 4 chips in the second round.

For the experiments, we use a vectorized form of CFR that applies regret updates to each information set consistent with the public information. The first vector variants were introduced in~\cite{Johanson12PCS}, and have been used in DeepStack and Libratus~\cite{Moravcik17DeepStack,Brown17Libratus}. See the appendix for more detail on the implementation. 
Baseline average values $\bar{\hat{u}}^b_i(I,a)$ used a decay factor of $\alpha = 0.5$. 
We used a uniform sampling in all our experiments, $\xi(I,a) = \frac{1}{|A(I)|}$.

We also consider the best case performance of our algorithm by using the oracle baseline. It uses baseline values of the true counterfactual values.
We also experiment with and without CFR+, demonstrating that our technique allows 
the CFR+ to be for the first time efficiently used with sampling.

\subsection{Convergence} 
We compared  MCCFR, MCCFR+, \gls{vrcfr}, \gls{vrcfr}+, and \gls{vrcfr}+ with the oracle baseline, see Fig.~\ref{fig:exploitability}. The variance-reduced VR-MCCFR and VR-MCCFR+ variants converge significantly faster than plain MCCFR. Moreover, the speedup grows as the baseline improves during the computation. A similar trend is shown by both \gls{vrcfr} and \gls{vrcfr}+, see Fig.~\ref{fig:speedup}. MCCFR needs hundreds of millions of iterations to reach the same exploitability as VR-MCCFR+ achieves in one million iterations: a 250-times speedup. VR-MCCFR+ with the oracle baseline significantly outperforms VR-MCCFR+ at the start of the  computation, but as time progresses and the learned baseline improves, the difference shrinks. After one million iterations, exploitability of VR-MCCFR+ with a learned baseline approaches the exploitability of VR-MCCFR+ with the oracle baseline. This oracle baseline result gives a bound on the gains we can get by constructing better learned baselines.

\begin{figure}[!ht]
\includegraphics[width=0.99\hsize]{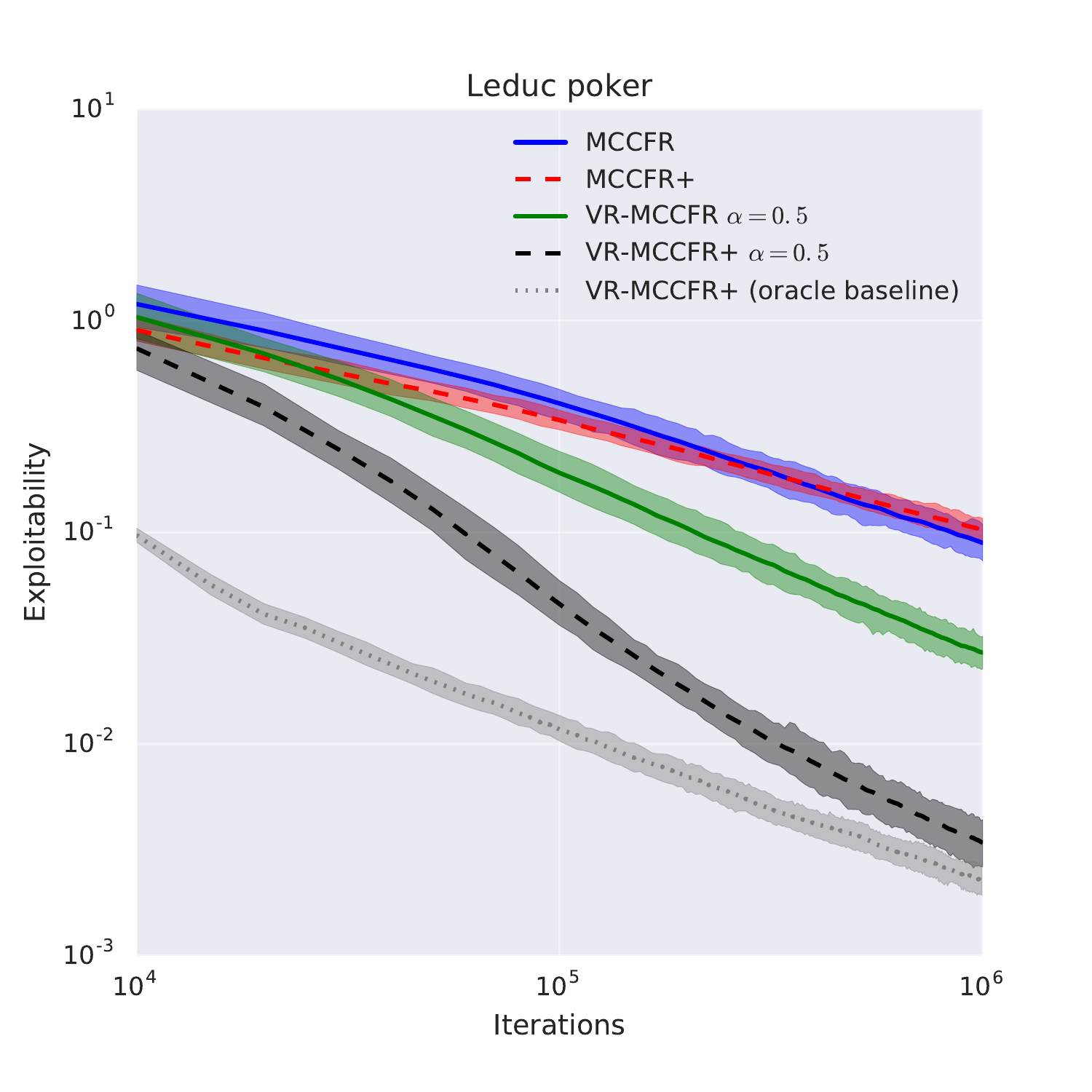}
\caption{Convergence of exploitability for different MCCFR variants on logarithmic scale. \gls{vrcfr} converges substantially faster than plain MCCFR. \gls{vrcfr}+ bring roughly two orders of magnitude speedup. \gls{vrcfr}+ with oracle baseline (actual true values are used as baselines) is used as a bound for \gls{vrcfr}'s performace to show possible room for improvement. When run for $10^6$ iterations \gls{vrcfr}+ approaches performance of the oracle version. The ribbons show 5th and 95th percentile over 100 runs.}
\label{fig:exploitability}
\end{figure}

\begin{figure}[!ht]
\includegraphics[width=0.99\hsize]{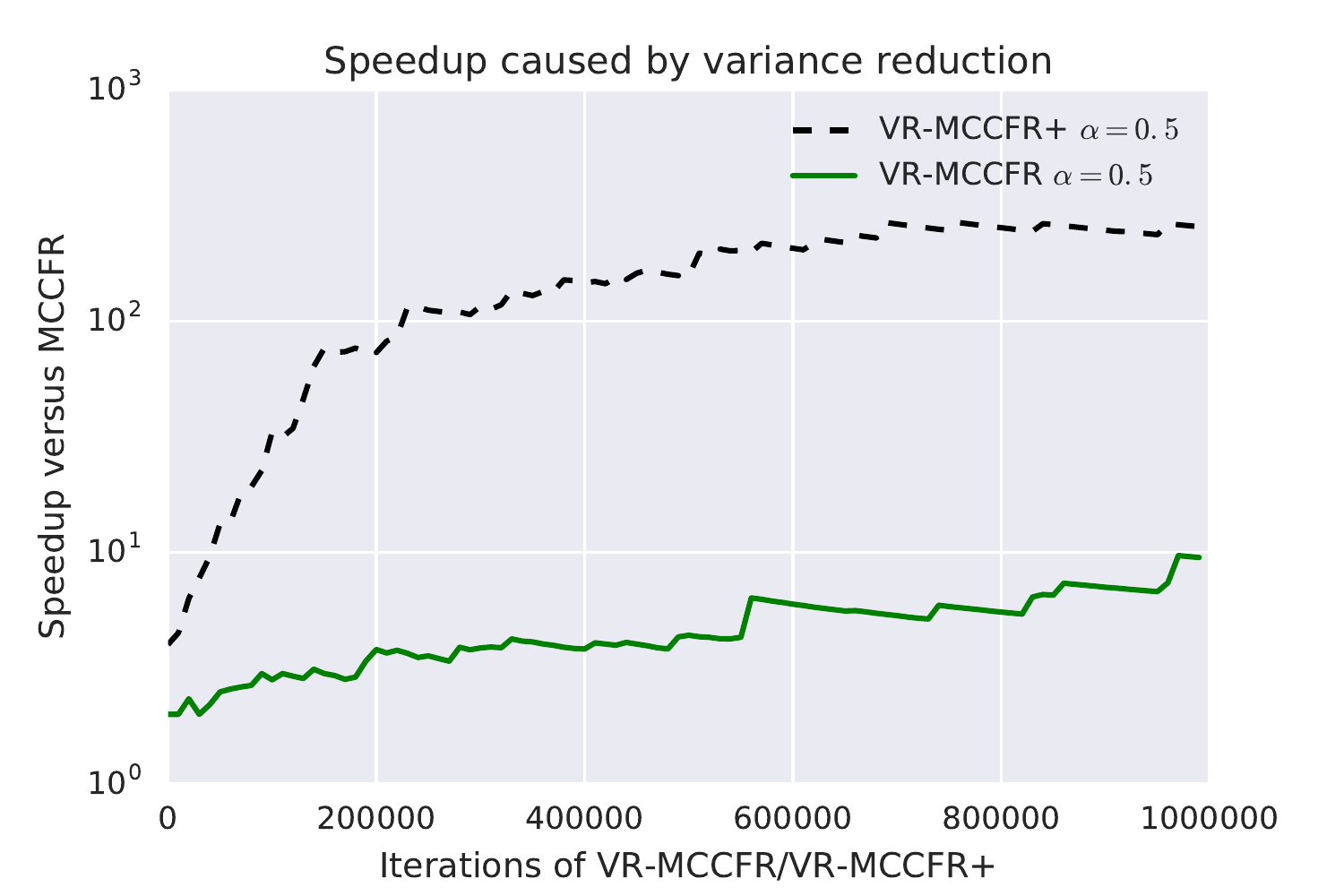}
\caption{Speedup of \gls{vrcfr} and \gls{vrcfr}+ compared to plain MCCFR. Y-axis show how many times more iterations are required by MCCFR to reach the same exploitability as \gls{vrcfr} or \gls{vrcfr}+.  
}
\label{fig:speedup}
\end{figure}

\subsection{Observed Variance}

To verify that the observed speedup of the technique is due to variance reduction, we experimentally observed variance of counterfactual value estimates for MCCFR+ and MCCFR, see Fig.~\ref{fig:variance-diagram}. We did that by sampling 1000 alternative trajectories for all visited information sets, with each trajectory sampling a different estimate of the counterfactual value. While the variance of value estimates in the plain algorithm  seems to be more or less constant, the variance of VR-MCCFR and VR-MCCFR+ value estimates is lower, and continues to decrease as more iterations are run. This confirms that the combination of baseline and bootstrapping is reducing variance, which implies better performance given the connection between variance and MCCFR's performance (Theorem~\ref{theorem:gibson}).

\begin{figure}[!ht]
\includegraphics[width=0.99\hsize]{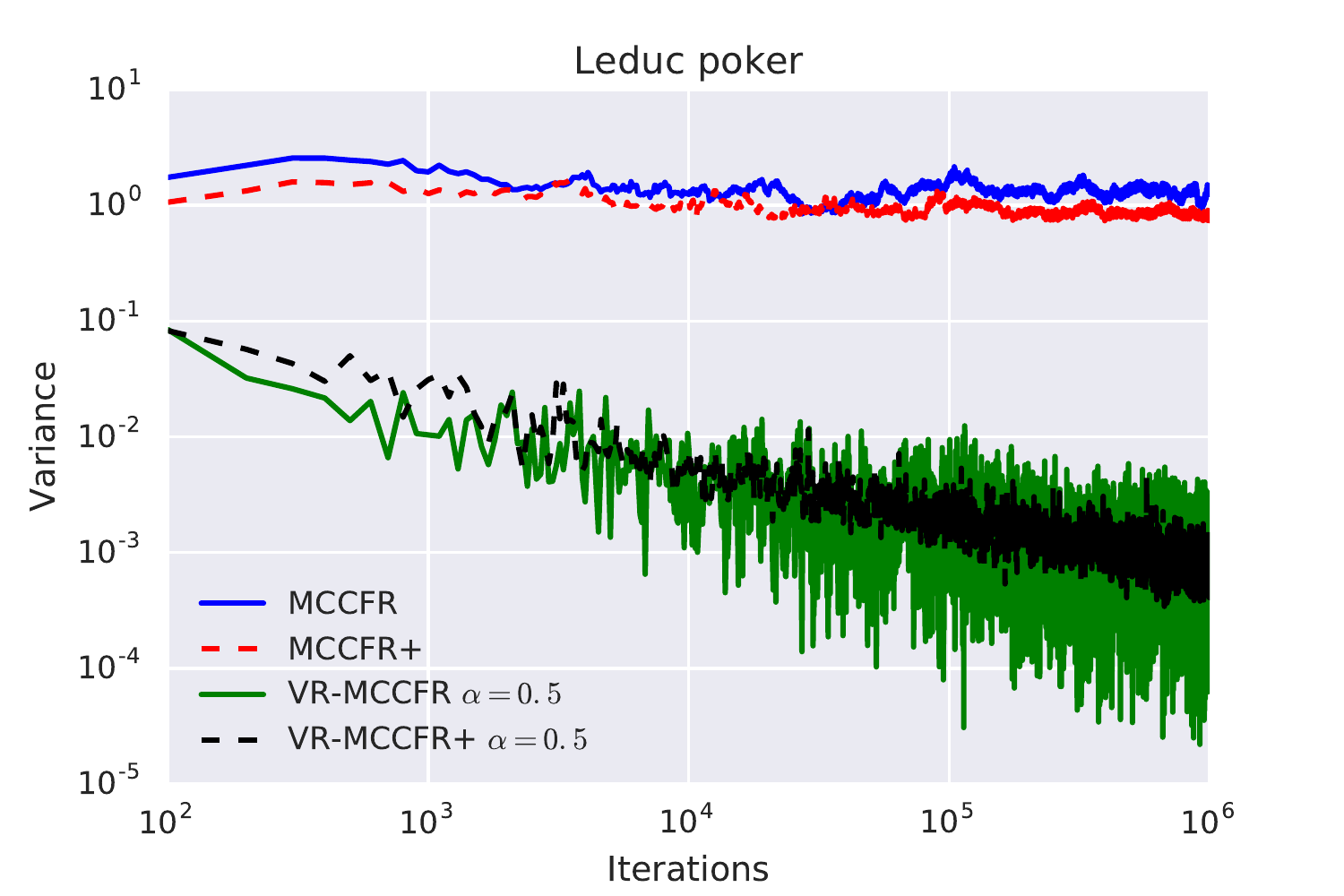}
\caption{Variance of counterfactual values in \gls{vrcfr} and plain MCCFR with both regret matching and regret matching+. The curves were smoothed by computing moving average over a sliding window of 100 iterations. }
\label{fig:variance-diagram}
\end{figure}

\subsection{Evaluation of Bootstrapping and Baseline Dependence on Actions}

Recent work that evaluates action-dependent baselines in RL \cite{tucker2018mirage}, shows that there is often no real advantage compared to baselines that depend just on the state.  It is also not common to bootstrap the value estimates in  RL. Since  VR-MCCFR uses both of these techniques it is natural to explore the contribution of each idea. We  compared four VR-MCCFR+ variants: with or without bootstrapping and with baseline that is state or state-action dependant, see Fig.~\ref{fig:exploitability-detailed}. The conclusion is that the improvement in the performance is very small unless we use both bootstrapping and an action-dependant baseline.

\begin{figure}[!ht]
\includegraphics[width=0.99\hsize]{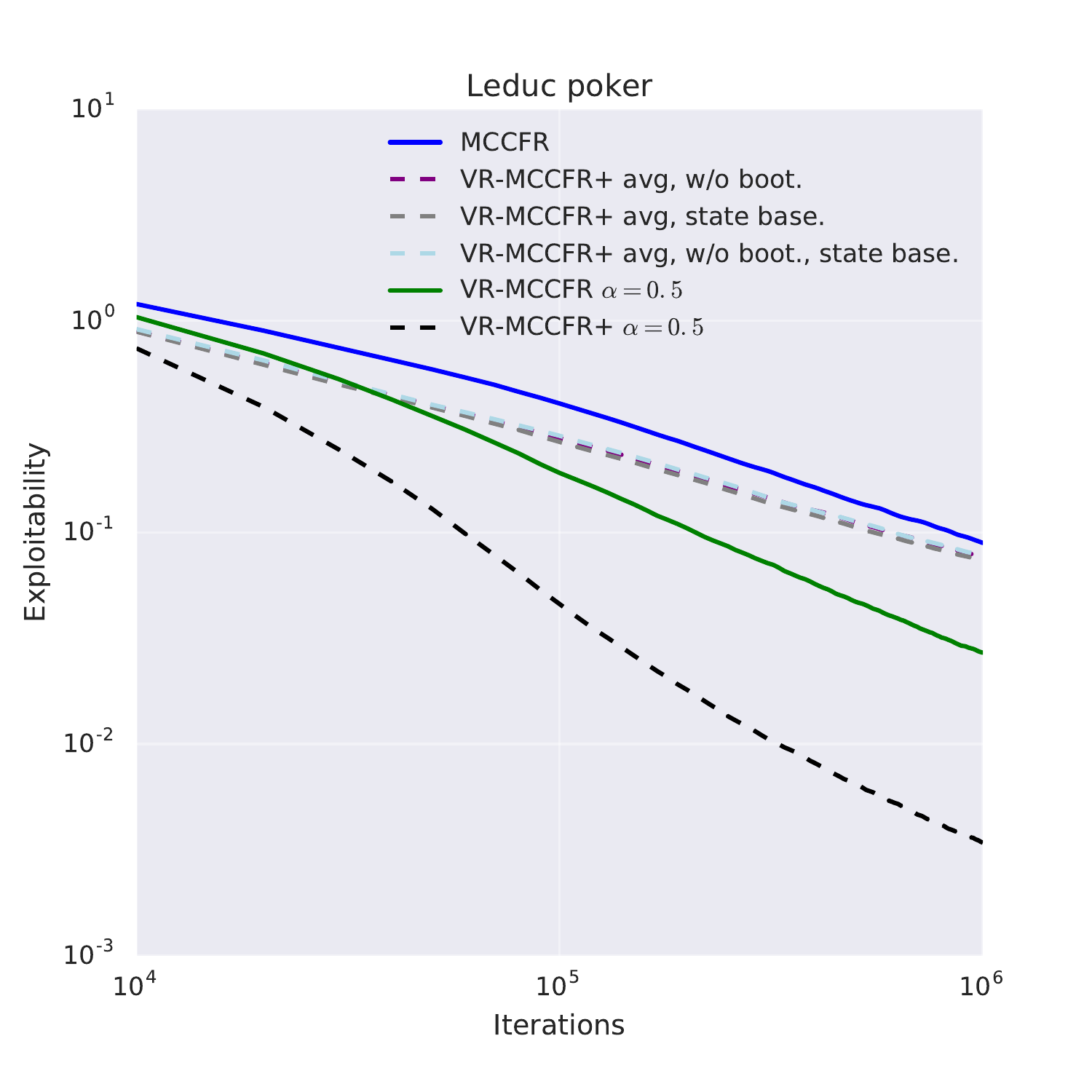}
\caption{Detailed comparison of different \gls{vrcfr} variants on logarithmic scale. The curves for MCCFR, \gls{vrcfr} and \gls{vrcfr}+ are the same as in the previous plot, the other lines show how the algorithm performs when using state baselines instead of state-action baselines, and without bootstrapping. All of these reduced variants perform better than plain MCCFR, however they are worse than full \gls{vrcfr}. This ablation study shows that the combination of all \gls{vrcfr} features is important for final performance.
}
\label{fig:exploitability-detailed}
\end{figure}

\section{Conclusions}

We have presented a new technique for variance reduction for Monte Carlo counterfactual regret minimization. This technique has close connections to existing RL methods of state and state-action baselines. 
In contrast to RL environments, our experiments in imperfect information games suggest that state-action baselines are superior to state baselines. Using this technique, we show that empirical variance is indeed reduced, speeding up the convergence by an order of magnitude.
The decreased variance allows for the first time CFR+ to be used with sampling, bringing the speedup to two orders of magnitude.

\bibliographystyle{named}
\bibliography{paper}

\newpage

\appendix
{\Large {\bf Appendices}}

\section{MCCFR and MCCFR+ comparison}

While it is known in that MCCFR+ is outperformed by MCCFR \cite{burch2017time}, we are not aware on any explicit comparison of these two algorithms in literature. Fig.~\ref{fig:mccfr_vs_mccfr_plus} shows experimental evaluation of these two techniques on Leduc poker.

\begin{figure}[!ht]
\includegraphics[width=0.99\hsize]{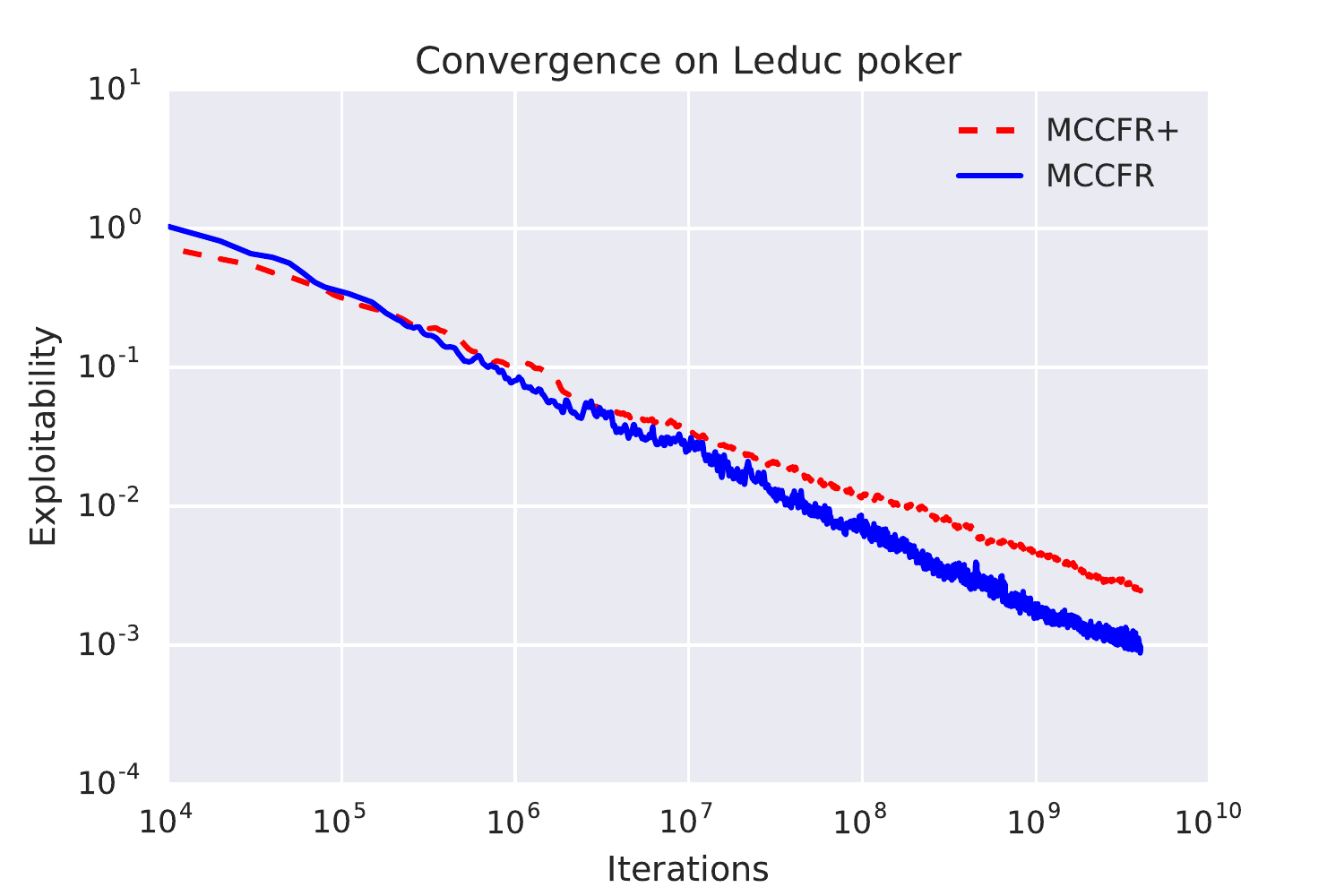}
\caption{Convergence of MCCFR and MCCFR+  on logarithmic scale. For the first $10^6$ iterations, MCCFR+ performs similllary to the MCCFR. After approximately $10^7$ iterations, the difference in favor of MCCFR starts to be visible and the gap in exploitability widens as the number of iterations grows. }
\label{fig:mccfr_vs_mccfr_plus}
\end{figure}

\section{Vector Form of CFR}

The first appearance of the vector form was presented in~\cite{Johanson11rgbr}. In this paper, the best response computation, needed to compute exploitability, was sped-up by re-defining the computation using the notion of a public tree. At the heart of a public tree is the notion of a \defword{public state} which contains a set of information sets whose histories are consistent with the public information revealed so far~\cite[Definition 2]{Johanson11rgbr}. This allowed the method to compute quantities for all information sets consistent with a public state at once (stored in vectors) and operations to compute them could be vectorized during a traversal of the public tree. There are also game-specific optimizations that could be applied at leaf nodes to asymptotically reduce the total computation necessary.

A similar construction was used in several sampling variants introduced in~\cite{Johanson12PCS}. Here, instead of computing necessary for best response, counterfactual values were vectorized and stored instead. The paper describes several ways to sample at various types of chance nodes (ones which reveal public information, or private information to each player), but the concept of a vectorized form of CFR was general. In fact, a vector form of vanilla CFR is possible in any game: when traversing down the tree, these vectors store the probability of reaching each information set (called a {\it range} in \cite{Moravcik17DeepStack}) and return vectors of counterfactual values. Both DeepStack and Libratus used vector forms of CFR and CFR+ in No-Limit poker.

For the MCCFR variants in this paper, the idea is the same as the previous sample variants. For any sequence of public actions, we concurrently maintain and update all information sets consistent with the sequence of public actions. For example in Leduc poker, six trajectories per player are maintained which all share the same sequence of public actions.

The main difference in our implementation is that baselines are kept as vectors at each public state, each representing a baseline for the information sets corresponding to the public state. Also, the average values tracked are counterfactual and normalized by the range. So, for example in Leduc, for five information sets in some public state, $(I_1, I_2, \ldots, I_5)$, quantity tracked by the baseline at this public state for action $a$ is:
\[
\frac{\hat{v}_i^b(\sigma, I_k, a)}{\sum_{k'}{\pi^{\sigma}_{opp}(I^{opp}_{k'})}},
\]
where $\pi^{\sigma}_{opp}$ is the reach probability of the opponent only (excluding chance), and $I^{opp}$ refers to the augmented information set belonging to the opponent at $I$.
Then, when using the baseline values to compute the modified counterfactual values, we need to multiply them by the current $\sum_{k'} \pi_{opp}^{\sigma}(I_{k'}^{opp})$ to get the baseline values under the current strategy $\sigma$.

\section{Proofs}

\subsection{Proof of Lemma~\ref{lemma:unbiasied}}

\begin{align*}
\bE[\hat{v}^b(\sigma, I,a)] & = \bE[\hat{v}_i(\sigma, I, a)] - \bE[\hat{b}_i(I, a)] + \bE[b_i(I, a)]\\
& = v_i(\sigma, I, a) - b_i(I, a) + b_i(I, a)\\
& = v_i(\sigma, I, a).~~~~~~~~~~~~~~~~~~~~~~~~~~~~~~~~~~~~~~~~~~~~~~~~~~~~~~~~~~~~\qed
\end{align*}

\subsection{Proof of Lemma~\ref{lemma:bootstrapped-unbiased}}

We begin by proving a few supporting lemmas regarding local expectations over actions at specific histories:
\begin{lemma}
\label{lemma:remove_baseline}
Given some $h \in \cH$, for any $z \in \cZ$ generated by sampling $\xi : \cH \mapsto \cA$ and all actions $a$, $\bE_{z \sim \xi}[\hat{u}_i^b(\sigma,h,a|z)] = \sum_{z, ha \sqsubseteq z} q(z)\hat{u}_i^b(\sigma,ha|z)/\xi(h,a)$:
\end{lemma}
\begin{proof}
$\hat{u}_i^b(\sigma,h,a|z)$ has three cases, from which we get $\bE_{z \sim \xi}[\hat{u}_i^b(\sigma,h,a|z)]$
\begin{eqnarray*}
& = & \sum_{z, ha \sqsubseteq z} q(z) \\
&   & \left( b_i(I_i(h),a) + \frac{-b_i(I_i(h),a) + \hat{u}_i^b(\sigma,ha|z)}{\xi(h,a)} \right) \\
& + & \sum_{z, h \sqsubset z, ha \not \sqsubseteq z} q(z)(b_i(I_i(h),a)) \\
& + & \sum_{h  \not  \sqsubseteq z} 0 \\
& = & \sum_{z, ha \sqsubseteq z} q(z)\hat{u}_i^b(\sigma,ha|z)/\xi(h,a) \\
& + & (q(ha)-q(ha)/\xi(h,a))b_i(I_i(h),a) \\
& + & q(h)(1-\xi(h,a))b_i(I_i(h),a) \\
& = & \sum_{z, ha \sqsubseteq z} q(z)\hat{u}_i^b(\sigma,ha|z)/\xi(h,a)
\end{eqnarray*}
\end{proof}

\begin{lemma}
\label{lemma:exp-basline-u}
Given some $h \in \cH$, for any $z \in \cZ$ generated by sampling $\xi : \cH \mapsto \cA$, the local baseline-enhanced estimate is an unbiased estimate of expected values for all actions $a$:
\begin{eqnarray*}
\bE_{z \sim \xi}[ \hat{u}_i^b(\sigma, h, a | z)] 
    & = & \bE_{z \sim \xi}[ \hat{u}_i(\sigma, h, a | z)].
\end{eqnarray*}
\end{lemma}
\begin{proof}
We prove this by induction on the maximum distance from $ha$ to any terminal. The base case is $ha \in \cZ$. $\bE_{z \sim \xi}[\hat{u}_i^b(\sigma, h, a | z)]$
\begin{eqnarray*}
& = & \sum_{z, ha \sqsubseteq z} q(z)\hat{u}_i^b(\sigma,ha|z)/\xi(h,a) ~~~~\mbox{by Lemma~\ref{lemma:remove_baseline}} \\
& = & \sum_{z, ha \sqsubseteq z} q(z)\hat{u}_i(\sigma,ha|z)/\xi(h,a) ~~~~\mbox{by Eq.~\ref{eq:bootstrapped-ub-history}} \\
& = & \bE_{z \sim \xi}[\hat{u}_i(\sigma, h, a | z)] ~~~~\mbox{by Eq.~\ref{eq:bootstrapped-u},~\ref{eq:bootstrapped-u-history}}
\end{eqnarray*}

Now assume for $i \ge 0$ that the lemma property holds for all $h'a'$ that are at most $j \le i$ steps from a terminal.  Consider history $ha$ being $i+1$ steps from some terminal, which implies that $ha \not\in \cZ$. We have $\bE_{z \sim \xi}[\hat{u}_i^b(\sigma, h, a | z)]$
\begin{eqnarray*}
& = & \sum_{z, ha \sqsubseteq z} q(z)\hat{u}_i^b(\sigma,ha|z)/\xi(h,a) ~~~~\mbox{by Lemma~\ref{lemma:remove_baseline}} \\
& = & \sum_{z, ha \sqsubseteq z} q(z) \sum_{a'}\sigma(ha,a')\hat{u}_i^b(\sigma,ha,a'|z)/\xi(h,a) \\
&   & \mbox{by Eq.~\ref{eq:bootstrapped-ub-history}} \\
& = & \sum_{z, ha \sqsubseteq z} q(z) \sum_{a'}\sigma(ha,a')\hat{u}_i(\sigma,ha,a'|z)/\xi(h,a) \\
&   & \mbox{by assumption} \\
& = & \sum_{z, ha \sqsubseteq z} q(z)\hat{u}_i(\sigma,ha|z)/\xi(h,a) ~~~~\mbox{by Eq.~\ref{eq:bootstrapped-u-history}} \\
& = & \bE_{z \sim \xi}[\hat{u}_i(\sigma, h, a | z)] ~~~~\mbox{by Eq.~\ref{eq:bootstrapped-u}}
\end{eqnarray*}
The lemma property holds for distance $i+1$, and so by induction the property holds for all $h$ and $a$.
\end{proof}

\begin{lemma}
\label{lemma:bootrapped-v-sampled-cfv}
Given some $h \in \cH$, for any $z \in \cZ$ generated by sampling $\xi : \cH \mapsto \cA$ and for all actions $a$, the local baseline-enhanced estimate is an unbiased estimate of the original sampled counterfactual value: $\bE_{z \sim \xi}[\hat{v}_i^b(\sigma, I_i(h), a | z)] = \bE_{z \sim \xi}[\tilde{v}_i(\sigma, I_i(h), a | z)]$.
\end{lemma}
\begin{proof}
First, $\bE_{z \sim \xi}[\hat{v}_i^b(\sigma, I_i(h), a | z)]$
\begin{eqnarray*}
& = & \bE_{z \sim \xi}\left[\frac{\pi_{-i}^{\sigma}(h)}{q(h)} \hat{u}_i^b(\sigma, h, a | z)\right] ~~~~\mbox{by Eq.~\ref{eq:bootstrapped-vb}}\\
& = & \frac{\pi_{-i}^{\sigma}(h)}{q(h)} \bE_{z \sim \xi}[ \hat{u}_i^b(\sigma, h, a | z)]\\
& = & \frac{\pi_{-i}^{\sigma}(h)}{q(h)} \bE_{z \sim \xi}[\hat{u}_i(\sigma, h, a | z)] ~~~~\mbox{by Lemma~\ref{lemma:exp-basline-u}}\\
& = & \bE_{z \sim \xi}[\tilde{v}_i(\sigma, I_i(h), a | z)] ~~~~\mbox{by Eq.~\ref{eq:sampled-cfv}, \ref{eq:bootstrapped-u}}.
\end{eqnarray*}
\end{proof}
\begin{proof}[Proof of Lemma 2]
The proof now follows directly:\\
$\bE_{z \sim \xi}[\hat{v}_i^b(\sigma, I, a | z)]$
\begin{eqnarray*}
& = &  \bE_{z \sim \xi}[\tilde{v}_i(\sigma, I, a | z)]~~~~\mbox{by Lemma~\ref{lemma:bootrapped-v-sampled-cfv}}\\
& = & v_i(\sigma, I, a)~~~~\mbox{by~\cite[Lemma 1]{Lanctot09mccfr}}.\\
\end{eqnarray*}
\end{proof}

\subsection{Proof of Lemma~\ref{lemma:zero-variance}}

We start by proving that given an oracle baseline, the baseline-enhanced expected value is always equal to the true expected value, and therefore has zero variance.

\begin{lemma}
\label{lemma:zero-variance-u}
Using an oracle baseline defined over histories, $b_i^*(h,a) = u^{\sigma}_i(ha)$, then for all $z$ such that $h \sqsubseteq z$, $\hat{u}^{b^*}_i(\sigma, h, a|z) = u_i^{\sigma}(ha)$.
\end{lemma}
\begin{proof}
Similar to above, we prove this by induction on the maximum distance from $ha$ to $z$. The base case is $ha \in \cZ$. By assumption $h \sqsubseteq z$ so we have $\hat{u}_i^{b^*}(\sigma,h,a|z)$
\begin{eqnarray*}
& = &
 \left\{ \begin{array}{ll}
 b_i^*(h, a) + \frac{\hat{u}_i^{b^*}(\sigma, ha | z) - b_i^*(h, a)}{\xi(h,a)} & \mbox{if $ha = z$} \\
 b_i^*(h, a) & \mbox{otherwise} \\
\end{array} \right. \\
&   & \mbox{by Eq.~\ref{eq:bootstrapped-ub}} \\
& = &
 \left\{ \begin{array}{ll}
 u_i^\sigma(ha) + \frac{u_i^\sigma(ha) - u_i^\sigma(ha)}{\xi(h,a)} & \mbox{if $ha = z$} \\
 u_i^\sigma(ha) & \mbox{otherwise} \\
\end{array} \right. \\
&   & \mbox{by Eq.~\ref{eq:bootstrapped-ub-history} and definition of $b_i^*(h,a)$} \\
& = & u_i^\sigma(ha)
\end{eqnarray*}

Now assume for $i \ge 0$ that the lemma property holds for all $h'a'$ that are at most $j \le i$ steps from a terminal. Consider history $ha$ being $i + 1$ steps from some terminal, which implies $ha \not \in \cZ$. We have
\begin{eqnarray}
\label{eq:oracle_baseline_history}
\hat{u}_i^{b^*}(\sigma,ha|z) = u_i^\sigma(ha)
\end{eqnarray}
because 
$\hat{u}_i^{b^*}(\sigma,ha|z)$
\begin{eqnarray*}
& = & \sum_{a'} \sigma(ha,a')\hat{u}_i^{b^*}(\sigma,ha,a'|z) ~~~~\mbox{by Eq.~\ref{eq:bootstrapped-ub-history}} \\
& = & \sum_{a'} \sigma(ha,a')u_i^\sigma(haa') ~~~~\mbox{by assumption} \\
& = & u_i^\sigma(ha) ~~~~\mbox{by definition of $u_i^\sigma$}
\end{eqnarray*}

We now look at $\hat{u}_i^{b^*}(\sigma,h,a|z)$
\begin{eqnarray*}
& = &
 \left\{ \begin{array}{ll}
 u_i^\sigma(ha) + \frac{\hat{u}^{b^*}_i(\sigma, ha | z) - u_i^\sigma(ha)}{\xi(h,a)} & \mbox{if $ha \sqsubset z$} \\
 u_i^\sigma(ha) & \mbox{otherwise} \\
\end{array} \right. \\
&   & \mbox{by Eq.~\ref{eq:bootstrapped-ub} and definition of $b_i^*(h,a)$} \\
& = &
 \left\{ \begin{array}{ll}
 u_i^\sigma(ha) + \frac{u_i^\sigma(ha) - u_i^\sigma(ha)}{\xi(h,a)} & \mbox{if $ha \sqsubset z$} \\
 u_i^\sigma(ha) & \mbox{otherwise} \\
\end{array} \right. \\
&   & \mbox{by Eq.~\ref{eq:oracle_baseline_history}} \\
& = & u_i^\sigma(ha)
\end{eqnarray*}
The lemma property holds for distance $i+1$, and so by induction the property holds for all $h$ and $a$.
\end{proof}

\begin{proof}[Proof of Lemma~\ref{lemma:zero-variance}]
Given $z$ such that $h \sqsubseteq z$, we have $\hat{v}_i^*(\sigma,h,a|z)$
\begin{eqnarray*}
& = & \frac{\pi_{-i}^\sigma(h)}{q(h)}\hat{u}_i^{b^*}(\sigma,h,a|z) ~~~~\mbox{by Eq.~\ref{eq:bootstrapped-vb}} \\
& = & \frac{\pi_{-i}^\sigma(h)}{q(h)}u_i^\sigma(ha) ~~~~\mbox{by Lemma~\ref{lemma:zero-variance-u}} \\
\end{eqnarray*}
None of the terms above depend on $z$, and so we have $\Var_{h,z \sim \xi, h \in I, h \sqsubseteq z}[\hat{v}_i^*(\sigma,h,a|z)] = 0$. Note as well that $\pi_{-i}^\sigma(h)u_i^\sigma(ha)$ corresponds to the terms in the summation of Equation~\ref{eq:action-dep-cfv}, so abusing notation, we have $\hat{v}_i^*(\sigma,h,a|z)=v_i(\sigma,h,a)/q(h)$: the counterfactual value of taking action $a$ at $h$, with an importance sampling weight to correct for the likelihood of reaching $h$.
\end{proof}

In MCCFR, the optimal baseline $b^*$ is not known, as it would require traversing the entire tree, taking away any advantages of sampling. However, $b^*$ can be approximated (learned online), which motivates the choice for tracking its average value presented in the main part of the paper.

\section{Kuhn Example}

In this section, we present a step-by-step example of one iteration of the algorithm on Kuhn poker~\cite{wiki:kuhn}. Kuhn poker is a simplified version of poker with three cards and is therefore suitable for demonstration purposes. Table~\ref{tab:kuhn-example} show forward pass of \gls{vrcfr} algorithm, Table~\ref{tab:kuhn-example-backward} shows backward pass.

\newcommand{\ke}[1]{\raisebox{-.5\height}{\includegraphics[page=#1]{kuhn-b.pdf}}}

\newcommand{\myline}{\cline{2-7}}

\newcommand{\notused}{\cellcolor{gray!20}}

\begin{table*}[t]

\centering
\begin{tabular}{|c|c|c|c|c|c|c|}
\hline
& \multicolumn{6}{c|}{Forward pass} \\ \myline
& $h$  & Game tree trajectory  & $\pi_{-1}^\sigma(h)$   & $q(h)$  & $I_1 = I_1(h)$  & $I_2 = I_2(h)$             \\ 
& \scriptsize{History} & & \scriptsize{Reach prob.} & \scriptsize{Sampling prob.} & \scriptsize{Infoset for Pl1} & \scriptsize{Infoset for Pl2} \\ \hline \hline
\multirow{5}{*}{\begin{minipage}{1cm} \ke{10} \end{minipage}}  & $\emptyset$ & \ke{5}        & 1                    & 1                        & \notused $\emptyset$    & \notused $\emptyset$       \\ 
\myline
& K     & \ke{6}        & $\frac{1}{3}$        & $\frac{1}{3}$      & K     & ?      \\ \myline
& KQ    & \ke{7}       & $\frac{1}{6}$        & $\frac{1}{6}$       & K?    & ?Q      \\ \myline
& KQB   & \ke{8}    & $\frac{1}{6}$        & $\frac{1}{12}$         & K?B   & ?QB       \\ \myline
& KQBC  & \ke{9}    & $\frac{1}{24}$       & $\frac{1}{24}$         & \notused K?BC  & \notused ?QBC      \\ \hline
\end{tabular}

\caption{Detailed example of updates computed for player 1 in Kuhn poker during forward pass of the algorithm. Backward pass that uses these values is shown in Table~\ref{tab:kuhn-example-backward}. In our representation history $h$ is a concatenation of all public and private actions. The game tree trajectory column shows the path in  the game tree that was sampled. Solid arrows denote sampled actions while dashed arrows show other available actions, all actions have their probability under current strategy $\sigma$ next to them. The sampled history in this case is: chance deals (K)ing to player 1, chance deals (Q)ueen to player 2, player 1 (B)ets, player 2 (C)alls. We will use shorter notation $KQBC$ to refer to this history. For each history $h$ reach probability $\pi_{-1}^\sigma(h)$ shows how likely the history is reached when player 1 plays in a way to get to this history. The sampling probabilities $q(h)$ are computed following sampling policy $\xi$ which is uniform in this case, i.e. for each history all available actions have the same probability that they will be sampled. The last two columns show augmented information sets for each player in each history. For example for player 1 history KQB is represented by information set K?B since he does not know what card was dealt to PLAYER 2. Light gray background marks cells where the values are well defined however they are not used in our example update for player 1.}
\label{tab:kuhn-example}
\end{table*}

\newcommand{\uKQBC}{\begin{tabular}{@{}c@{}} $\mathbf{\hat{u_1}^b({ \scriptstyle \sigma, h, | z})} = u_1(h) $

\\  $=2$

\end{tabular}}


\newcommand{\uKQB}{\begin{tabular}{@{}c@{}} $\mathbf{\hat{u}^b_1(\sigma, h | z)} =$
\\ $\sum_{a} \sigma(h,a) \hat{u}^b_1(\sigma, h, a | z)$
\\ $=\frac{3}{4} * (-2) + \frac{1}{4} * 3$
\\ $=-\frac{3}{4}$ \end{tabular}}

\newcommand{\uKQBc}{\begin{tabular}{@{}c@{}} $\mathbf{\hat{u_1}^b({ \scriptstyle \sigma, h,C | z})} =$
\\ $\frac{\hat{u_1}^b({ \scriptstyle \sigma, hC | z})  - b (\scriptstyle I_1, c)}{\xi{({  \scriptstyle h,C})}} + b({ \scriptstyle  I_1, C})$

\\  $=\frac{2-1}{\frac{1}{2}} + 1$

\\ $=3$

\end{tabular}}

\newcommand{\uKQBf}{\begin{tabular}{@{}c@{}} $\mathbf{\hat{u_1}^b({ \scriptstyle \sigma, h,F | z})} =b({ \scriptstyle  I_1, F})$

\\  $=-2$

\end{tabular}}

\newcommand{\uKQBa}{\begin{tabular}{@{}c@{}} \uKQBc

\\ \uKQBf \end{tabular}}



\newcommand{\uKQ}{\begin{tabular}{@{}c@{}} $\mathbf{\hat{u}^b_1(\sigma, h | z)}=$ 
\\ $\sum_{a} \sigma(h,a) \hat{u}^b_1(\sigma,h, a|z)$
\\ $=\frac{1}{3} * (-1) + \frac{2}{3} * (-2)$
\\ $=-\frac{5}{3}$ \end{tabular}}

\newcommand{\uKQb}{\begin{tabular}{@{}c@{}} $\mathbf{\hat{u_1}^b({  \sigma,  h,B | z})}=$
\\ $\frac{\hat{u_1}^b({ \scriptstyle  \sigma,  hB | z})  - b (\scriptstyle I_1, B)}{\xi{({  \scriptstyle h,B})}} + b({ \scriptstyle  I_1, B})$

\\  $=\frac{-\frac{3}{4}-0.5}{\frac{1}{2}} + 0.5$

\\ $=-2$

\end{tabular}}

\newcommand{\uKQc}{\begin{tabular}{@{}c@{}} $\mathbf{\hat{u_1}^b({ \scriptstyle \sigma, h,C  |z})} = b({ \scriptstyle  I_1, C})$

\\  $=-1$

\end{tabular}}

\newcommand{\uKQa}{\begin{tabular}{@{}c@{}} \uKQb

\\ \uKQc \end{tabular}}


\newcommand{\vKQb}{\begin{tabular}{@{}c@{}} $\mathbf{\hat{v}^b_1(\sigma, I_1, B | z)} =$
\\ $\frac{\pi^{\sigma}_{-1}(h)}{q(h)} \hat{u}^b_1(\sigma, h, B | z)$

\\  $=\frac{\frac{1}{6}}{\frac{1}{6}} * (-2)$

\\ $=-2$

\end{tabular}}

\newcommand{\vKQc}{\begin{tabular}{@{}c@{}} $\mathbf{\hat{v}^b_1(\sigma, I_1, C | z)} =$
\\ $\frac{\pi^{\sigma}_{-1}(h)}{q(h)} \hat{u}^b_1(\sigma, h, C | z)$

\\  $=\frac{\frac{1}{6}}{\frac{1}{6}} * (-1)$

\\ $=-1$

\end{tabular}}

\newcommand{\vKQa}{\begin{tabular}{@{}c@{}} \vKQb

\\ \vKQc \end{tabular}}

\newcommand{\mylineb}{\cline{2-6}}
\newcommand{\keb}[1]{\raisebox{-.5\height}{\includegraphics[page=#1]{kuhn-b.pdf}}}

\begin{table*}[t]

\centering
\begin{tabular}{|c|c|c||c|c|c|}
\hline
& & \multicolumn{4}{|c|}{Backward pass}\\ 
& $h$          & Game tree trajectory & $\hat{u_1}^b(\sigma, h,a|z)$         & $\hat{u_1}^b(\sigma, h|z)$   & $\hat{v_1}^b(\sigma, I_1, a | z)$ \\ 

& \scriptsize{History} & & \scriptsize{Sampled corrected history-action utility} & \scriptsize{Sampled corrected history utility} & \scriptsize{Sampled corrected cf-value}  \\ \hline

Def.&            &         & Eq.~\ref{eq:bootstrapped-ub}  &  Eq.~\ref{eq:bootstrapped-ub-history}        &  Eq.~\ref{eq:bootstrapped-vb}           \\ \hline \hline

\multirow{5}{*}{\begin{minipage}{1cm} \ke{11} \end{minipage}} & $\emptyset$         & \keb{5}          &   \notused  & \notused     & \notused    \\ 
\mylineb
& K            & \keb{6}          & \notused  & \notused & \notused       \\ \mylineb
& KQ          & \keb{7} & \uKQa      &  \uKQ  \notused   & \vKQa      \\ \mylineb
& KQB      &   \keb{8} & \uKQBa  & \uKQB     &  \notused       \\ \mylineb
& KQBC & \keb{9}    & \notused     &  \uKQBC   &    \notused   \\ \hline

\end{tabular}
\caption{The backward pass starts by evaluating utility of the terminal history: $\hat{u_1}^b(\sigma,KQBC|KQBC) = +2$ since player 1 has (K)ing which is better card than opponent's (Q)ueen. In the next step computation updates values for history $KQB$. Expected baseline corrected history-action value $\hat{u_1}^b(\sigma,KQB,Call|KQBC)$ is computed based on current sample and then used together with $\hat{u_1}^b(\sigma,KQB,Fold|KQBC)$ to compute $\hat{u_1}^b(\sigma,KQB|KQBC)$. When updating values for history KQ baseline corrected sampled counterfactual values are computed based on just updated $\hat{u_1}^b(\sigma,KQ,Bet|KQBC)$ for the sampled Bet action and on a baseline value $\hat{u_1}^b(\sigma,KQ,Check|KQBC)$ for Check action that was not sampled. Reach probability $\pi_{-1}^\sigma(KQ)$ and sampling probability $q(KQ)$ that are also needed to compute counterfactual-values $\hat{v_1}^b(\sigma, K?, a | KQBC)$ were already computed in the forward pass. The counterfactual values are then used to compute actions' regrets (Eq.~\ref{eq:regret}) which is not shown in the table. Values in cell with light gray background are not used in computation of $\hat{v_1}^b(\sigma, K?, a | KQBC)$.}
\label{tab:kuhn-example-backward}
\end{table*}

\end{document}